\documentclass{llncs}
\usepackage[T1]{fontenc}
\usepackage[latin9]{inputenc}
\usepackage{array}
\usepackage{verbatim}
\usepackage{amsmath}
\usepackage{amssymb}
\usepackage{graphicx}
\usepackage{epstopdf}
\usepackage{todonotes}
\usepackage{subfigure}
\usepackage{times}
\usepackage[defblank]{paralist}
\makeatletter
\providecommand{\tabularnewline}{\\}
\usepackage{amsmath}
\usepackage{txfonts}
\pdfmapfile{+txfonts.map}
\makeatother

\begin{document}
\title{Verification of Query Completeness over Processes [Extended Version]}
\author{Simon Razniewski \and Marco Montali \and Werner Nutt}
\institute{Free University of Bozen-Bolzano\\
Dominikanerplatz 3\\
39100 Bozen-Bolzano\\
\{razniewski,montali,nutt\}$@$inf.unibz.it}
\maketitle
\begin{abstract}
Data completeness is an essential aspect of data quality, and has in
turn a huge impact on the effective management of companies.
For example, statistics are computed and audits are conducted in
companies by implicitly placing the strong assumption that the
analysed data are complete. 
In this work, we are interested in studying the problem of
completeness of data produced by business processes, to the aim of
automatically assessing whether a given database query can be answered
with complete information in a certain state of the process. 
We formalize so-called \emph{quality-aware processes} that create data
in the real world and store it in the company's information system possibly at a
later point.
%
%
%
We then show how
one can check the completeness of database queries in a certain
state of the process or after the execution of a sequence of actions, by leveraging on query
containment, a well-studied problem in database theory.
\end{abstract}
\global\long\def\newmacroname{}

\global\long\def\fc{f_{\C}}

\global\long\def\uk{\mathrm{uk}}
 %

\global\long\def\na{\mathrm{N/A}}
 %

\global\long\def\const{\mathrm{const}}

\global\long\def\codd{\mathrm{Codd}}

\global\long\def\sql{\mathrm{SQL}}

\global\long\def\can{\mathrm{can}}

\global\long\def\sred{\mathrm{sred}}

\global\long\def\complsred#1{\mathit{Compl^{\sred}}(#1)}

\global\long\def\complb#1{\mathit{Compl^{b}}(#1)}

\global\long\def\compls#1{\mathit{Compl^{s}}(#1)}

\global\long\def\complstar#1#2{\mathit{Compl}^{#2}(#1)}

\global\long\def\smixed{\mathrm{s+}}

\newcommand{\Null}{\textsl{null}\xspace}
\newcommand{\Nulls}{\textsl{null}s\xspace}
\newcommand{\pdb}{{\mathcal D\xspace}\xspace}
\newcommand{\D}{{\pdb}}
\newcommand{\id}[1]{#1^i}
\newcommand{\av}[1]{#1^a}

\newcommand{\epos}[2]{\textit{EPos}(#1,#2)}   

\newcommand{\etal}{et al.}
\newcommand{\satisfies}{\models}




\newcommand{\RQ}{\ensuremath{\L_{\rm{RQ}}}}
\newcommand{\CQ}{\ensuremath{\L_{\rm{CQ}}}}
\newcommand{\LRQ}{\ensuremath{\L_{\rm{LRQ}}}}
\newcommand{\LCQ}{\ensuremath{\L_{\rm{LCQ}}}}

\renewcommand{\L}{{\mathcal L}}							

\newcommand{\compl}[1]{\textit{Compl}($#1$)}
\newcommand{\complpi}[1]{\textit{Compl}_{\pi}\ensuremath{(#1)}}


\newcommand{\localComp}{table completeness\xspace}
\newcommand{\LocalComp}{Table completeness\xspace}
\newcommand{\LC}{TC\xspace}

\newcommand{\LCQC}{\ensuremath{\mathsf{TC\text{-}QC}}}
\newcommand{\LCLC}{\ensuremath{\mathsf{TC\text{-}TC}}}
\newcommand{\QCQC}{\ensuremath{\mathsf{QC\text{-}QC}}}

\newcommand{\PTIME}{\ensuremath{\mathsf{PTIME}}}
\newcommand{\CONP}{\ensuremath{\mathsf{coNP}}}
\newcommand{\NP}{\ensuremath{\mathsf{NP}}}

\newcommand{\Cont}{\ensuremath{\mathsf{Cont}}}
\newcommand{\ContU}{\ensuremath{\mathsf{ContU}}}


\newcommand{\eset}{\emptyset}
\newcommand\bigset[1]{ \Bigl\{ #1 \Bigr\} }   
\newcommand\bigmid{\ \Big|\ }

\newcommand{\bag}[1]{\{\hspace{-0.13em}|\, #1 \,|\hspace{-0.13em}\}}   

\newcommand{\incl}{\subseteq}           
\newcommand{\incls}{\supseteq}          

\newcommand{\col}{\colon}

\newcommand{\angles}[1]{\langle#1\rangle}       

\newcommand{\dd}[2]{#1_1,\ldots,#1_{#2}}      

\newcommand{\quotes}[1]{\lq\lq#1\rq\rq}         
\newcommand{\wrt}{w.r.t.\ }                        
\newcommand{\WLOG}{wlog\xspace}                     

\newcommand{\Sum}{\ensuremath{{\sf sum}}}
\newcommand{\Count}{\ensuremath{{\sf count}}}
\newcommand{\Max}{\ensuremath{{\sf max}}}
\newcommand{\Min}{\ensuremath{{\sf min}}}

\newcommand\core[1]{\mathring{#1}}


\newcommand{\qif}{\,{:}{-}\,}
\newcommand{\union}{\cup}
\newcommand{\dom}{\mathit{dom}}
\newcommand{\var}{\mathit{Var}}
\newcommand{\lit}[1]{L_{#1}}
\newcommand{\cpred}[1]{V_{#1}}
\newcommand{\concond}{G}
\newcommand{\da}{\av D}
\newcommand{\di}{\id D}
\newcommand{\lcstmt}{C}
\newcommand{\cplstmt}[2]{#1 {\, \dot{\subseteq} \, } #2}
\newcommand{\Compl}[1]{\mathit{Compl}{(}#1{)}}
\newcommand{\Complb}[1]{\mathit{Compl}^b($#1$)}
\newcommand{\Compls}[1]{\mathit{Compl}^s($#1$)}
\newcommand{\Complsred}[1]{\mathit{Compl}^{sred}($#1$)}

\newcommand{\qcompl}[1]{\mathit{Compl}{(}#1{)}}
\newcommand{\cplset}{{\mathcal C}}
\newcommand{\C}{{\mathcal C}}                       
\newcommand{\Q}{{\mathcal Q}}                       
\newcommand{\F}{{\mathcal F}}                       
\renewcommand{\S}{{\mathcal S}}                       
\newcommand{\cplhat}{\hat C}
\newcommand{\chk}{\check}
\renewcommand{\implies}{\rightarrow}
\renewcommand{\And}{\wedge}
\newcommand{\Or}{\vee}
\newcommand{\findom}[3]{\mathit{Dom}(#1,#2,#3)} 

\newcommand{\true}{\textit{true}}
\newcommand{\false}{\textit{false}}
\newcommand{\determines}{\rightarrow \! \! \! \! \! \rightarrow} 
\newcommand{\dotequiv}{\ \dot{\equiv} \ } 
\newcommand{\Wlog}{W.l.o.g.\ }
\newcommand{\ie}{i.\ e.\ ,}
\newcommand{\piptwo}{\Pi^P_2}
\newcommand{\sigptwo}{\Sigma^P_2}

\newcommand{\queryset}{{\mathcal Q}} 
\newcommand{\schemaconstraintset}{{\mathcal F}} 
\newcommand{\scset}{\schemaconstraintset} 
\newcommand{\qset}{\queryset}
\newcommand{\val}{\upsilon}
\newcommand{\viewset}{{\mathcal V}}

\newcommand{\tpl}[1]{\bar{#1}}				
\newcommand{\tplsub}[2]{{\bar{#1}}_{#2}}				
\newcommand{\query}[2]{#1 \qif #2} 		
\newcommand{\aufz}[2]{#1_1,\ldots,#1_{#2}}    
\newcommand{\set}[1]{\{\,#1\,\}}
\newcommand{\Onlyif}{\lq\lq$\Rightarrow$\rq\rq\ \ }   
\newcommand{\If}{\lq\lq$\Leftarrow$\rq\rq\ \ }        
\newcommand{\eat}[1]{}
\newcommand{\modelsms}{\models \! \! ^{\mathit{M \! \! S}}}		  

\newcommand{\joininquotes}{\mbox{ ``$\bowtie$"}}      


\global\long\def\pdb{\mathcal{D}}

\global\long\def\domby{\preceq}

\global\long\def\tc#1#2#3{\mathit{Compl}(#1;\,#2;\,#3)}

\global\long\def\query#1#2{#1\qif#2}




\newcommand{\student}{\texttt{student}\xspace}
\newcommand{\sid}{\texttt{sid}\xspace}
\newcommand{\name}{\texttt{name}\xspace}
\newcommand{\level}{\texttt{level}\xspace}
\newcommand{\code}{\texttt{code}\xspace}
\newcommand{\hometown}{\texttt{hometown}\xspace}
\newcommand{\class}{\texttt{class}\xspace}
\newcommand{\formTeacher}{\texttt{formTeacher}\xspace}
\newcommand{\viceFormTeacher}{\texttt{viceFormTeacher}\xspace}
\newcommand{\profile}{\texttt{profile}\xspace}

\newcommand{\customer}{\texttt{customer}}
\newcommand{\street}{\texttt{street}}
\newcommand{\city}{\texttt{city}}
\newcommand{\supplier}{\texttt{supplier}}
\newcommand{\partner}{\texttt{partner}}
\newcommand{\longTerm}{\texttt{long\_term}}
\newcommand{\contact}{\texttt{contact}}
\newcommand{\iD}{\texttt{id}}

\global\long\def\pdb{\mathcal{D}}

\global\long\def\domby{\preceq}

\global\long\def\dominates{\succeq}

\global\long\def\tc#1#2#3{\mathit{Compl}(#1;#2;#3)}

\global\long\def\query#1#2{#1\qif#2}

\global\long\def\compl#1{\mathrm{Compl}(#1)}

\global\long\def\tpl#1{\bar{#1}}

\global\long\def\ctle{\mathbf{E}}

\global\long\def\ctlf{\mathbf{F}}

\global\long\def\ctla{\mathbf{A}}

\global\long\def\ctlg{\mathbf{G}}

\global\long\def\muexists{<\negthickspace\negthinspace-\negthickspace\negthinspace>}

\global\long\def\muall{[\negthinspace-\negthinspace]}

\textit{}\global\long\def\statedescriptor{\mathrm{stateDescriptor}}

\global\long\def\mulp{\mu\mathcal{L_{\mathit{P}}}}

\global\long\def\mula{\mu\mathcal{L_{\mathit{A}}}}

\global\long\def\live{\mathrm{Live}}

\global\long\def\chase{\mathrm{chase}}

\global\long\def\b{\mathrm{bag}}

\global\long\def\s{\mathrm{set}}

\global\long\def\D{\mathcal{D}}

\global\long\def\S{\mathcal{S}}

\global\long\def\P{\mathcal{P}}

\global\long\def\formcurcompl{\gamma^{\mathrm{cur}}}

\global\long\def\formdefcompl{\gamma^{\mathrm{def}}}

\global\long\def\id#1{#1^{\mathit{rw}}}

\global\long\def\av#1{#1^{\mathit{is}}}

\global\long\def\true{\mathrm{true}}

\global\long\def\formdefcomplcl{\formdefcompl_{C,L}}

\global\long\def\sarrown{\S^{\rightarrow n}}

\global\long\def\tcom#1#2{\mathrm{Compl}(#1;#2)}

\global\long\def\tcor#1#2{\mathrm{Corr}(#1;#2)}

\global\long\def\formcurcor{\epsilon^{\mathrm{cur}}}

\global\long\def\formdefcor{\epsilon^{\mathrm{def}}}

\global\long\def\ag{\mathbf{\mathbf{\boldsymbol{\mathbf{AG}}}}}

\global\long\def\ef{\mathbf{\mathbf{\boldsymbol{\mathbf{EF}}}}}

\global\long\def\tphi{T^{\phi}}

\global\long\def\intrcur{\mathrm{Intr}^{\mathrm{cur}}}

\global\long\def\intrdef{\mathrm{Intr}^{\mathrm{def}}}

\global\long\def\intr#1{\mathrm{Intr}(#1)}

\global\long\def\piptwo{\Pi_{2}^{P}}

\section{Introduction}

Data completeness is an important aspect of data quality. When data
is used in decision-making, it is important that the data is of good
quality, and in particular that it is complete. This is particularly
true in an enterprise setting. On the one hand, strategic decisions
are taken inside a company by relying on statistics and
business indicators such as KPIs. Obviously, this information is
useful only if it is reliable, and reliability, in turn, is strictly
related to quality and, more specifically, to completeness.  

Consider for example the school information system of the autonomous
province of Bolzano in Italy, which triggered the research
included in this paper. Such an information system stores data
about schools, enrolments, students and teachers. When statistics are
computed for the enrolments in a given school, e.g., to decide the amount of teachers
needed for the following academic year, it is of utmost importance
that the involved data are complete, i.e., that the required
information stored in the information system is aligned with reality. 

Completeness of data is a key issue also in the context of
auditing. When a company is
evaluated to check whether its way of conducting business is in
accordance to the law and to audit assurance standards, part of the
external audit is dedicated to the analysis of the actual data. If
such data are incomplete w.r.t.~the queries issued during the audit,
then the obtained answers do not properly reflect the company's behaviour.

There has been plenty of work on fixing data quality issues, especially
for fixing incorrect data and for detecting duplicates \cite{hernandez:1998:data-cleansing,bilenko:2003:duplicate-detection}. However, some
data quality issues cannot be automatically fixed. This holds
in particular for incomplete data, as missing data cannot be corrected
inside a system, unless additional activities are introduced to acquire them.
In all these situations, it is then a mandatory requirement to (at least) detect
data quality issues, enabling informed decisions drawn with knowledge
about which data are complete and which not. 

The key question therefore is
how it is possible to obtain this completeness information.
There has been previous work on the assessment of data completeness
\cite{Razniewski:Nutt-Compl:of:Queries-VLDB11}, however this approach
left the question where completeness information come from largely
open. In this work, we argue that, in the common situation where the manipulation of data
inside the information system is driven by business processes, we
can leverage on such processes to infer information about data
completeness, provided that we suitably annotate the involved
activities with explicit information about the way they manipulate data.

A common source of data incompleteness in business processes is constituted by delays between real-world events
and their recording in an information system. This holds in particular
for scenarios where processes are carried out partially without
support of the information system. E.g., many legal events are considered valid as soon as they are signed on a sheet
of paper, but their recording in the information system could happen
much later in time. Consider again the example of the school
information system, in particular the enrolment of pupils in schools.
Parents enroll their children at the individual schools, and the enrolment
is valid as soon as both the parents and the school director sign
the enrolment form. However, the school secretary may record the information
from the sheets only later in the local database of the school, and even later submit
all the enrolment information to the central school administration,
which needs it to plan the assignment of teachers to
schools, and other management tasks.

In the BPM context, there have been attempts to model data quality
issues, like in \cite{bpmn-data-inaccuracy-modeling-weak,bpmn-data-quality:cappiello:caballero-weak,bagchi:2006:data-quality-and-business-process-modeling}.
However, these approaches mainly discussed general methodologies for
modelling data quality requirements in BPMN, but did not provide methods
to asses their fulfilment.
In this paper, we claim that process formalizations are an essential
source for learning about data completeness and show how data completeness
can be verified. In particular, our contributions are (1) to introduce
the idea of extracting information about data completeness from
processes manipulating the data, (2) to formalize processes that can
both interact with the real-world and record information about the
real-world in an information system, and (3) to show how completeness
can be verified over such processes, both at design and at execution time.

Our approach leverages on two assumptions related to how the data
manipulation and the process control-flow are captured. From the data
point of view, we leverage on annotations that suitably mediate
between expressiveness and tractability. More specifically, we rely on
annotations modeling that new information of a given type is acquired
in the real world, or that some information present in the real world
is stored into the information system. We do not explicitly consider
the evolution of specific values for the data, as incorporating full-fledged
data without any restriction would immediately make our problem
undecidable, being simple reachability queries undecidable in such a
rich setting \cite{DDHV11,BCDDF11,BCDDM13}. From the control-flow point of
view, we are completely orthogonal to process specification
languages. In particular, we design our data completeness algorithms
over (labeled) transition systems, a well-established mathematical structure to
represent the execution traces that can by produced according to the
control-flow dependencies of the
(business) process model of interest. Consequently, our approach can in principle
be applied to any process modeling language, with the proviso of annotating
the involved activities. We are in particular interested in providing
automated reasoning facilities to answer whether a given query can be
answered with complete information given a target state or a
sequence of activities.


The rest of this paper is divided as follows. In Section 2, we discuss
the scenario of the school enrolment data in the province of Bozen/Bolzano
in detail. In Section 3, we discuss our formal approach, introducing
quality-aware transition systems, process activity annotations used to
capture the semantics of activities that interact
with the real world and with an information system, and 
properties of query completeness over such systems. In Section 4,
we discuss how query completeness can be verified over such systems
at design time and at runtime, how query completeness can be refined and what the complexity of deciding query completeness is.

\section{Example Scenario}

Consider the example of the enrollment to schools in the province
of Bolzano. Parents can submit enrollment requests for their child
to any school they want until the 1st of March. Schools then decide
which pupils to accept, and parents have to choose one of the schools
in which their child is accepted. Since in May the school administration
wants to start planning the allocation of teachers to schools and take
further decisions (such as the opening and closing of school branches
and schools) they require the schools to process the enrollments
and to enter them in the central school information system before the
15th of April. 

A particular feature of this process is that it is partly carried out
with pen and paper, and partly in front of a computer, interacting
with an underlying school information system. Consequently, 
the information system does not often contain all the information
that hold in the real world, and is therefore incomplete. E.g.,
while an
enrollment is legally already valid when the enrollment sheet is signed,
this information is visible in the information system only when the
secretary enters it into a computerised form.

A BPMN diagram sketching the main phases of this process is shown in Fig.~\ref{figure:enrolment-BPMN-advanced},
while a simple UML diagram of (a fragment of) the school domain is
reported in Fig.~\ref{figure:ER-diagram-school-world}. 
These diagrams abstractly summarise the school domain from the point
of view of the central administration. Concretely, each school 
implements a specific, local version of the enrolment process, relying
on its own domain conceptual model. The data collected on a
per-school basis are then transferred into a central information
system managed by the central administration, which refines the
conceptual model of Fig.~\ref{figure:ER-diagram-school-world}. In the
following, we will assume that such an information system represents
information about children and the class they belong to by means of a
$pupil(pname,class,sname)$ relation, where $pname$ is the name of an
enrolled child, $class$ is the class to which the pupil belongs, and
$sname$ is the name of the corresponding school. 
\begin{figure}[t]
\centering
\includegraphics[scale=1]{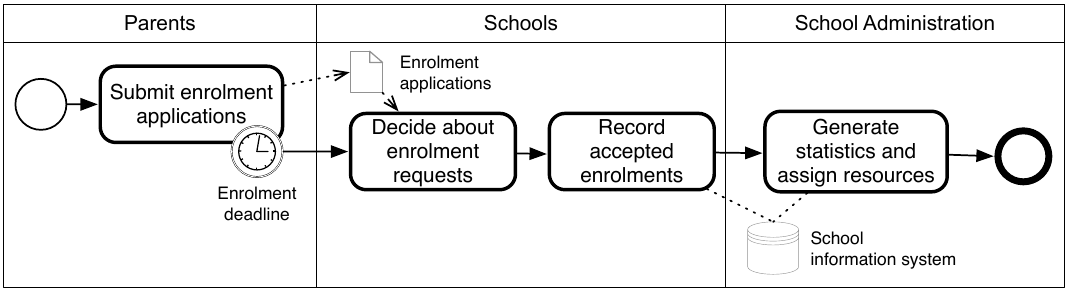}\caption{BPMN
  diagram of the main phases of the school enrollment process}
\label{figure:enrolment-BPMN-advanced}
\end{figure}

\begin{figure}[t]
\centering
\includegraphics[scale=1]{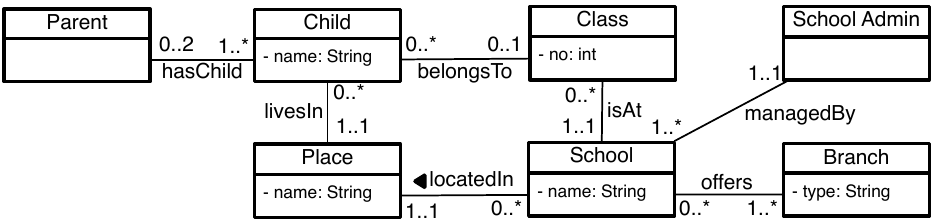}
\caption{UML diagram capturing a fragment of the school domain}
\label{figure:ER-diagram-school-world}
\end{figure}

When using the statistics about the enrollments as compiled in the
beginning of May, the school administration is highly interested in
having correct statistical information, which in turn requires that the underlying data about
the enrollments must be complete. Since the data is generated during
the enrollment process, this gives rise to several questions about
such a process.
The first question is whether the process is generally designed correctly,
that is, whether the enrollments present in the information system are really
complete at the time they publish their statistics, or whether it
is still possible to submit valid enrollments by the time the statistics
are published. We call this problem the \emph{design-time verification}.

A second question is to find whether the number of enrollments
in a certain school branch is already complete before the 15th of
April, that is, when the schools are still allowed to submit enrolments
(i.e., when there are school that still have not completed the second activity in the school lane
of Fig.~\ref{figure:enrolment-BPMN-advanced}), which could be the
case when some schools submitted all their enrollments but others
did not. In specific cases the number can be complete already, when
the schools that submitted their data are all the schools that offer
the branch. We call this problem the \emph{run-time verification}.

A third question is to learn on a finer-grained level about the completeness
of statistics, when they are not generally complete. When a statistic
consists not of a single number but of a set of values (e.g. enrollments
per school), it is interesting to know for which schools the number
is already complete. We call this the \emph{dimension
analysis}.

\section{Formalization}

We want to formalize processes such as the one in Fig.~\ref{figure:enrolment-BPMN-advanced},
which operate both over data in the real-world (pen\&paper) and record
information about the real world in an information system. We therefore
first introduce real-world databases and information system databases, and show then how transition systems, which represent possible process executions, can be annotated with effects for interacting with the real-world or the information system database.

\subsection{Real-world Databases and Information System Databases}

We assume an ordered, dense set of constants $\dom$ and a fixed set $\Sigma$ of relations. A database instance
is a finite set of facts in $\Sigma$ over $\dom$. As there exists both the
real world and the information system, in the following we model this
with two databases: $\id D$ called the real-world database, which
describes the information that holds in the real world, and $\av D$,
called the information system database, which captures the information
that is stored in the information system. We assume that the stored
information is always a subset of the real-world information. Thus,
processes actually operate over pairs $(\id D,\av D)$ of real-world
database and information system database. In the following, we will
focus on processes that create data in the real world and copy parts
of the data into the information system, possibly delayed. \global\long\def\John{\mathit{John}}
\global\long\def\Mary{\mathit{Mary}}
\global\long\def\pupil{\mathit{pupil}}
\global\long\def\req{\mathit{request}}
\global\long\def\enr{\mathit{enrolled}}
\global\long\def\test{\mathit{test}}
\global\long\def\Hoferschool{\mathit{HoferSchool}}
\global\long\def\DaVincischool{\mathit{DaVinciSchool}}
\global\long\def\livesin{\mathit{livesIn}}
\global\long\def\Bolzano{\mathit{Bolzano}}
\global\long\def\Bob{\mathit{Bob}}
\global\long\def\Alice{\mathit{Alice}}
\global\long\def\Merano{\mathit{Merano}}
\global\long\def\coNP{\mathrm{coNP}}
\global\long\def\NP{\mathrm{NP}}

\begin{example}
Consider that in the real world, there are the two pupils John and
Mary enrolled in the classes 2 and 4 at the Hofer School, while the
school has so far only processed the enrollment of John in their IT
system. Additionally it holds in the real world that John and Alice live in Bolzano
and Bob lives in the city of Merano. The real-world database $\id D$
would then be $\{\pupil(\John,2,\Hoferschool),\pupil(\Mary,4,\Hoferschool)$,\linebreak{}$\livesin(\John,\Bolzano)$,
 $\livesin(\Bob,\Merano),\livesin(Alice,\Bolzano)\}$ while the information
system database would be $\{\pupil(\John,2,\Hoferschool)\}$.\label{example:incomplete-db}
\end{example}
Where it is not clear from the context, we annotate atoms with the database
they belong to, so, e.g., $\av{\pupil}(\John,4,\Hoferschool)$
 means that this fact is stored in the information system database.

\subsection{Query Completeness}

For planning purposes, the school administration is interested 
in figures such as the number of pupils per class, school, profile,
etc. Such figures can be extracted from relational databases via SQL queries using the 
COUNT keyword. In an SQL database with a table $\mathtt{pupil(name,class,school)}$,
a query asking for the number of students per school would be written
as:
\\[-3.2ex]
\begin{align} 
\begin{split} 
& \mbox{SELECT school, COUNT(*) as pupils\_nr}\\[-.8ex]
& \mbox{FROM pupil}\\[-.8ex]
& \mbox{GROUP BY school.}\\[-4ex]
\label{equation-query1}
\end{split} 
\end{align} In database theory, conjunctive queries were introduced to formalize
SQL queries. A \emph{conjunctive query} $Q$ is an expression of the form
$\query{Q(\tpl x)}{A_{1},\ldots,A_{n},M}$, where $\tpl x$ are called
the distinguished variables in the head of the query, $A_{1}$
to $A_{n}$ the atoms in the body of the query, and $M$ is a set of built-in comparisons \cite{foundations_of_dbs}. We denote the set of all variables that appear in a query $Q$ by $\var(Q)$. Common subclasses of conjunctive queries are linear conjunctive queries, that is, they do not contain a relational symbol twice, and relational conjunctive queries, that is, queries that do not use comparison predicates.
Conjunctive queries allow to formalize all single-block SQL queries,
i.e., queries of the form ``SELECT $\ldots$ FROM $\ldots$ WHERE $\ldots$''.
As a conjunctive query, the SQL query (\ref{equation-query1}) above
would be written as:
\begin{equation}
\query{Q_{p/s}(\mathit{schoolname},\mathit{count(name}))}{\pupil(\mathit{name,class,schoolname})}
\end{equation}

In the following, we assume that all queries are conjunctive queries. 
We now formalize query completeness over a pair of a
real-world database and an information system database. Intuitively,
if query completeness can be guaranteed, then this means that the query over the generally incomplete information system database
gives the same answer as it would give w.r.t.~the information that holds
in the real world. Query completeness is the key property that we
are interested in verifying.

A pair of databases $(\id D,\av D)$ satisfies\emph{ query completeness}
of a query $Q$, if $Q(\id D)=Q(\av D)$ holds. We then write $(\id D,\av D)\models\compl Q$.
\begin{example}
Consider the pair of databases $(\id D,\av D)$ from Example \ref{example:incomplete-db}
and the query $Q_{p/s}$ from above (2).
Then, $\compl{Q_{p/s}}$ does not hold over $(\id D,\av D)$ because
$Q(\id D)=\{(\Hoferschool,2)\}$ but 
$Q(\av D)=\{(\Hoferschool,1)\}$. 
A query for pupils in class 2 only, $\query{Q_{class2}(n)}{\pupil(n,2,s)}$, 
would be complete, because $Q(\id D)=Q(\av D)=\{\John\}$.
\end{example}

\subsection{Real-world Effects and Copy Effects}
\label{sec:effects}
\global\long\def\actrw{\mathrm{RA}}
\global\long\def\actdb{\mathrm{CA}}
\global\long\def\tgd#1#2{#1\rightarrow#2}
\global\long\def\re{\mathit{re}}
\global\long\def\ce{\mathit{ce}}
\global\long\def\CE{\mathrm{CE}}
\global\long\def\start{start}
\global\long\def\resident{\mathit{resident}}
\global\long\def\copyy{\mathit{copy}}
\global\long\def\Davinci{\mathit{DaVinci}}

We want to formalize the real-world effect of an enrollment action at the Hofer School,
where in principle, every pupil that has submitted an enrolment request before, is allowed to enroll in
the real world. We can formalize this using the following implication:
$\id{\pupil}(n,c,\Hoferschool)\leftsquigarrow\id{\req}(n,\Hoferschool)$, which should mean
that whenever someone is a pupil at the Hofer school now, he has submitted an enrolment request before.
Also, we want to formalize copy effects, for example where all pupils
in classes greater than 3 are stored in the database. This can be written
with the following implication: $\id{\pupil}(n,c,s),c>3\rightarrow\av{\pupil}(n,c,s)$,
which means that whenever someone is a pupil in a class with level greater
than three in the real world, then this fact is also stored in the information
system.

For annotating processes with information about data creation and
manipulation in the real world $\id D$ and in the information system
$\av D$, we use real-world effects and copy effects as annotations.
While their syntax is the same, their semantics is different. Formally,
a \emph{real-world effect} $r$ or a \emph{copy effect} $c$ is a
tuple $(R(\tpl x,\tpl y),G(\tpl x,\tpl z))$, where $R$ is an atom,
$G$ is a set of atoms and built-in comparisons and $\tpl x$, $\tpl y$ and $\tpl z$ are
sets of distinct variables. We call $G$ the \emph{guard} of the effect.
The effects $r$ and $c$ can be written as follows:
\begin{eqnarray*}
&&r:\ \id R(\tpl x,\tpl y)\leftsquigarrow\exists\tpl z \! : \ \id G(\tpl x,\tpl z)\\
&&c:\ \id R(\tpl x,\tpl y),\id G(\tpl x,\tpl z)\rightarrow\av R(\tpl x,\tpl y)
\end{eqnarray*}

Real-world effects can have variables $\tpl y$ on the left side that do not occur in the condition. These variables are not restricted and thus allow to introduce new values.

A pair of real-world databases $(\id{D_{1}},\id{D_{2}})$ \emph{conforms}
to a real-world effect $\id R(\tpl x,\tpl y)\leftsquigarrow\exists\tpl z:\ \id G(\tpl x,\tpl z)$,
if for all facts $\id R(\tpl c_{1},\tpl c_{2})$ that are in $\id{D_{2}}$
but not in $\id{D_{1}}$ it holds that there exists a tuple of  constants $\tpl c_{3}$
such that the guard $\id G(\tpl c_{1},\tpl c_{3})$ is in $\id{D_{1}}$.
The pair of databases conforms to a set of real-world effects, if
each fact in $\id{D_{2}}\setminus\id{D_{1}}$ conforms to at least
one real-word effect. 

If for a real-world effect there does not exist any pair of databases $(D_1,D_2)$ with $D_2 \setminus D_1 \neq \emptyset$ that conforms to the effect, the effect is  called \emph{useless}. In the following we only consider real-world effects that are not useless.

The function $\copyy_{c}$ for a copy effect $c=\id R(\tpl x,\tpl y),\id G(\tpl x,\tpl z)\rightarrow\av R(\tpl x,\tpl y)$
over a real-world database $\id D$ returns the corresponding R-facts
for all the tuples that are in the answer of the query $\query{P_{c}(\tpl x,\tpl y)}{\id R(\tpl x,\tpl y),\id G(\tpl x,\tpl z)}$
over $\id D$. For a set of copy effects $\CE$, the function $\mathrm{copy_{\CE}}$
is defined by taking the union of the results of the individual copy
functions.
\begin{example}
Consider a real-world effect $r$
that allows to introduce persons living in Merano as pupils in classes higher than 3 in the real world, that is, 
$r=\id{\pupil}(n,c,s)\leftsquigarrow c>3,\livesin(n,\Merano)$
and a pair of real-world databases using the database $\id D$ from
Example \ref{example:incomplete-db} as 
 $(\id D,\id D\cup\{\id{\pupil}(\Bob,4,\Hoferschool)\}$. Then this
pair conforms to the real-world effect $r$, because the guard of
the only new fact $\id{\pupil}(\Bob,4,\Hoferschool)$ evaluates to
true: Bob lives in Merano and his class level is greater than 3. The
pair $(\id D,\id D\cup\{\id{\pupil}(\Alice,1,\Hoferschool)\}$ does
not conform to $r$, because Alice does not live in Merano, and also
because the class level is not greater than 3.

For the copy effect $c=\id{\pupil}(n,c,s),c>3\rightarrow\av{\pupil}(n,c,s)$,
which copies all pupils in classes greater equal 3, its output over
the real-world database in Example \ref{example:incomplete-db} would
be $\{\av{\pupil}(\Mary,4,\Hoferschool)\}$.

\end{example}

\subsection{Quality-Aware Transition Systems}
\global\long\def\qats{\bar{T}}

To capture the execution semantics of \emph{quality-aware processes},
we resort to (suitably annotated) labelled transition systems, a common way to describe the
semantics of concurrent processes by interleaving \cite{BaKG08}. This
makes our approach applicable for virtually every business process modelling
language equipped with a formal underlying transition semantics (such
as Petri nets or, directly, transition systems).

Formally, a \emph{(labelled) transition system} $T$ is a tuple $T=(S,s_{0},A,E)$,
where $S$ is a set of states, $s_{0}\in S$ is the initial state,
$A$ is a set of names of actions and $E\subseteq S\times A\times S$
is a set of edges labelled by actions from $A$. In the following,
we will annotate the actions of the transition systems with effects that describe
interaction with the real-world and the information system.
In particular, we introduce \emph{quality-aware transition systems}
(QATS) to capture the execution semantics of processes that change data both in the real world and
in the information system database. 

Formally, a \emph{quality-aware
transition system} $\qats$ is a tuple $\qats=(T,\re,\ce)$, where
$T$ is a transition system and $\re$ and $\ce$ are functions from
$A$ into the sets of all real-world effects and copy effects, which in
turn obey to the syntax and semantics defined in Sec.~\ref{sec:effects}. Note that transition systems and hence also QATS may contain cycles.

\begin{example}
Let us consider two specific schools, the Hofer School and the
Da Vinci School, and a (simplified version) of
their enrolment process, depicted in BPMN in
Fig.~\ref{fig:school-processes} (in parenthesis, we introduce compact
names for the activities, which will be used throughout the
example). As we will see, while the two processes are independent from
each other from the control-flow point of view (i.e., they run in
parallel), they eventually write information into the same table of
the central information system.

Let us first consider the Hofer School. In the first step, the requests are
processed with pen and paper, deciding which requests are accepted
and, for those, adding the signature of the school
director and finalising other bureaucratic issues. By using relation
$\id \req(n,\Hoferschool)$ to model the fact that a child named $n$ requests to be
enrolled at Hofer, and  $\id \pupil(n,1,\Hoferschool)$ to model that she is actually enrolled, the
activity \textsf{pH} is a real-world activity that can be annotated
with the real-world effect
$\id\pupil(n,1,\Hoferschool)\leftsquigarrow\id\req(n,\Hoferschool)$. In the second step, the
information about enrolled pupils is transferred to the central
information system by copying all real-world enrolments of the Hofer school. More specifically, the activity  \textsf{rH} can
be annotated with the copy effect 
$\id{\pupil}(n,1,\Hoferschool)\rightarrow\av{\pupil}(n,1,\Hoferschool)$.

Let us now focus on the Da Vinci School. Depending on the amount of
incoming requests, the school decides whether to directly process the
enrolments, or to do an entrance test for obtaining a ranking. In the
first case (activity \textsf{pD}), the activity mirrors that of the
Hofer school, and is annotated with the real-world effect
$\id\pupil(n,1,\Davinci)\leftsquigarrow\id\req(n,\Davinci)$. As for the test, the
activity \textsf{tD}  can be annotated with a real-world effect that
makes it possible to enrol only those children who passed the test:
$\id\pupil(n,1,\Davinci)\leftsquigarrow\id\req(n,\Davinci),\id{\test}(n,mark),mark{\ \geq\
}6$. Finally, the process terminates by properly transferring the
information about enrolments to the central administration, exactly as
done for the Hofer school. In particular, the activity \textsf{rD}  is
annotated with the copy effect
$\id\pupil(n,1,\Davinci)\rightarrow\av{\pupil}(n,1,\Davinci)$. Notice
that this effect feeds the same $\pupil$ relation of the central
information systems that is used by \textsf{rH}, but with a different value for the third column
(i.e., the school name).

Fig.~\ref{fig:school-qats} shows the QATS formalizing the execution
semantics of the parallel composition of the two processes (where
activities are properly annotated with the previously discussed
effects). Circles drawn in orange with solid line represent execution states where the information about pupils enrolled at the Hofer school is
complete. Circles in blue with double stroke represent execution states where completeness holds for pupils enrolled at the Da Vinci
school. At the final, sink state information about the
enrolled pupils is complete for both schools.

\begin{figure}[t]
\centering
\subfigure[]{\label{fig:school-processes}\includegraphics[width=.48\textwidth]{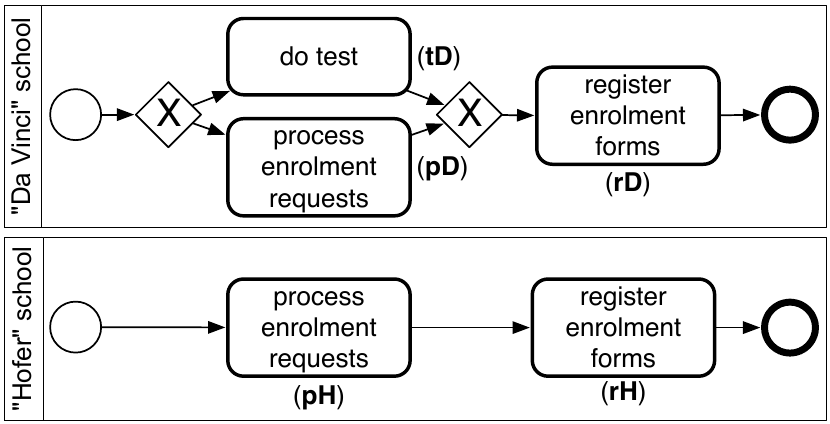}}
\subfigure[]{\label{fig:school-qats}\includegraphics[width=.48\textwidth]{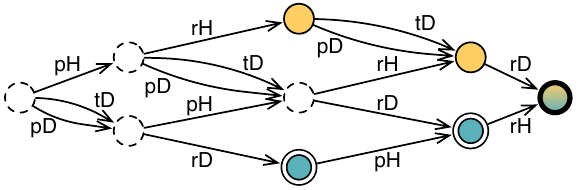}}
\caption{BPMN enrolment process of two schools, and the corresponding QATS}
\label{schools-bpmn-and-qats}
\end{figure}



\end{example}

\subsection{Paths and Action Sequences in QATSs}

Let $\qats=(T,\re,\ce)$ be a QATS. A \emph{path} $\pi$ in $\qats$
is a sequence $t_{1},\ldots,t_{n}$ of transitions such that $t_{i}=(s_{i-1},a_{i},s_{i})$
for all $i=1\ldots n$. An\emph{ action sequence} $\alpha$ is a sequence
$a_{1},\ldots,a_{m}$ of action names. Each path $\pi=t_{1},\ldots,t_{n}$
has also a \emph{corresponding} \emph{action sequence} $\alpha_{\pi}$
defined as $a_{1},\ldots,a_{n}$ . For a state $s$, the set \global\long\def\aseq{\mathit{Aseq}}
 $\aseq(s)$ is the set of the action sequences of all paths that
end in $s$.

Next we consider the semantics of action sequences. 
A \emph{development} of an action sequence $\alpha=a_{1},\ldots,a_{n}$
is a sequence $\id{D_{0}},\ldots,\id{D_{n}}$ of real-world databases
such that each pair $\id{(D_{j}},\id{D_{j+1}})$ conforms to the effects
$\re(\alpha_{j+1})$. Note that $\id{D_{0}}$ can be arbitrary. For each
development $\id{D_{0}},\ldots,\id{D_{n}}$, there exists a unique
trace $\av{D_{0}},\ldots,\av{D_{n}}$, which is a sequence of information
system databases $\av{D_{j}}$ defined as follows: 
\[
\av{D_{j}}=\begin{cases}
\id{D_{j}} & \mbox{if }j=0\\
\av{D_{j-1}}\cup\mathrm{copy}_{\CE(t_{j})}(\id{D_{j}}) & \mbox{otherwise.}
\end{cases}
\]
Note that $\av{D_{0}}=\id{D_{0}}$ does not introduce loss of generality
and is just a convention. To start with initially different databases,
one can just add an initial action that introduces data in all real-world
relations.

\subsection{Completeness over QATSs}

An action sequence $\alpha=a_{1},\ldots,a_{n}$ \emph{satisfies} query
completeness of a query $Q$, if for all developments of $\alpha$
it holds that $Q$ is complete over $\id{(D_{n}},\av{D_{n}})$, that
is, if $Q\id{(D_{n}})=Q(\av{D_{n}})$ holds.
A path $P$ in a QATS $\qats$ satisfies query completeness for $Q$,
if its corresponding action sequence satisfies it. 
A state $s$ in a QATS $\qats$ satisfies $\compl Q$, if all action sequences
in $\aseq(s)$ (the set of the action sequences of all paths that
end in $s$) satisfy $\compl Q$. We then write $s\models\compl Q$.
%
%
\begin{example}
Consider the QATS in Figure \ref{schools-bpmn-and-qats}(b) and recall that the action $\textsf{pH}$ is annotated with the real-world effect $\id\pupil(n,1,\Hoferschool)\leftsquigarrow\id\req(n,\Hoferschool)$, and action $\textsf{rH}$ with the copy effect $\id{\pupil}\!(n,1,\Hoferschool)\rightarrow\av{\pupil}\!(n,1,\Hoferschool)$.
A path $\pi=((s_{0},\textsf{pH},s_{1}),$
$(s_{1},\textsf{rH},s_{2}))$ has the corresponding
action sequence $(\textsf{pH},\textsf{rH})$. Its models are all sequences $(\id{D_{0}},\id{D_{1}},\id{D_{2}})$
of real-world databases (developments), where $\id{D_{1}}$ may contain
additional pupil facts at the Hofer school w.r.t.~$\id{D_{0}}$ because of the real-world
effect of action $a_{1}$, and $\id{D_{2}}=\id{D_{1}}$. Each such
development has a uniquely defined trace $(\av{D_{0}},\av{D_{1}},\av{D_{2}})$
where $\av{D_{0}}=\id{D_{0}}$ by definition, $\av{D_{1}}=\av{D_{0}}$
because no copy effect is happening in action $a_{1}$, and $\av{D_{2}}=\av{D_{1}}\cup copy_{\ce(a_{1})}(\id{D_{1}})$,
which means that all pupil facts from Hofer school that hold in the real-world database are
copied into the information system due to the effect of action $a_{1}$.
Thus, the state $s_{2}$ satisfies $\compl{Q_{\textit{Hofer}}}$ for a query $\query{Q_{\textit{Hofer}(n)}}{\pupil(n,c,\Hoferschool)}$, because
in all models of the action sequence the real-world pupils at the Hofer school are copied by the copy effect in action $\textsf{rH}$.
\end{example}

\section{Verifying Completeness over Processes}

In the following, we analyze how to check completeness in a state
of a QATS at design time, at runtime, and how to analyze the completeness
of an incomplete query in detail.

\subsection{Design-Time Verification}

When checking for query completeness at design time, we have to consider
all possible paths that lead to the state in which we want to check
completeness. We first analyze how to check completeness for a single
path, and then extend our results to sets of paths.

Given a query $\query{Q(\tpl z)}{R_{1}(\tpl{t}_{1}),\ldots,R_{n}(\tpl{t}_{n}),M},$
we say that a real-world effect $r$ is \emph{risky} w.r.t.~$Q$, if
there exists a pair of real-world databases $(\id{D_{1}},\id{D_{2}})$
that conforms to $r$ and where the query result changes, that is,
$Q(\id{D_{1}})\neq Q(\id{D_{2}})$. Intuitively, this means that real-world
database changes caused by $r$ can influence the query answer and lead
to incompleteness, if the changes are not copied into the information
system.
\begin{proposition}[Risky effects]
Let $r$ be the real-world effect $R(\tpl x, \tpl y)\leftsquigarrow G_{1}(\tpl x,\tpl z_{1})$, $Q$ be the query $\query{Q}{R_1(\tpl t_1),\ldots R_n(\tpl t_n),M}$ and $\tpl v = \var(Q)$. Then $r$ is risky wrt.~$Q$
 if and only if the following formula is satisfiable:
\[ 
G_{1}(\tpl x,\tpl z_{1}) 
\ \wedge \
\bigl(\bigwedge_{i=1\ldots n} \! R_{i}(\tpl{t}_{i}) \bigr)
\ \wedge \
 M 
\ \wedge \
\bigl( \bigvee_{R_{i}=R}(\tpl x, \tpl y)=\tpl{t}_{i}\bigr)
\]
\end{proposition}
\begin{fullversion}{
\begin{proof}
"$\Leftarrow :$"
If the formula is satisfied for some assignment $\delta$, this satisfying assignment directly yields an example showing that $r$ is risky wrt.\ $Q$ as follows: Suppose that the disjunct is satisfied for some $i=k$. Then we can construct databases $\id D_1$ and $\id D_2$ as $\id D_1 = G_1(\delta \tpl x, \delta \tpl z_1) \cup \set{ \bigwedge_{i=1\ldots n,i\neq k} R_i(\delta \tpl t_i)}$ and $\id D_2 = \id D_1 \cup \set{R_k(\delta \tpl t_k)}$.
Clearly, $(\id D_1 ,\id D_2)$ satisfies the effect $r$ because for the only additional fact $R_k(\delta \tpl t_k)$ in $\id D_2$, the condition $G_1$ is contained in $(\id D_1)$. But $Q(\id D_1)\neq Q(\id D_2)$ because with the new fact, a new valuation for the query is possible by mapping each atom to itself.

"$\Rightarrow :$" Holds by construction of the formula, which checks whether it is possible for $R$-facts to satisfy both $G_1$ and $Q$. Suppose $r$ is risky wrt.\ $Q$. Then there exists a pair of databases $(\id D_1 ,\id D_2)$ that satisfies $r$ and where $Q(\id D_1)\neq Q(\id D_2)$. Thus, all new facts in $\id D_2$ must conform to $G_1$ and some facts must also contribute to new evaluations of $Q$ that lead to $Q(\id D_1)\neq Q(\id D_2)$. Thus, each such facts implies the existence of a satisfying assignment for the formula. \qed
\end{proof}
}\end{fullversion}
\begin{example}
Consider the query $\query{Q(n)}{\pupil(n,c,s),\livesin(n,\Bolzano)}$
and the real-world effect $r_{1}=\pupil(n,c,s)\leftsquigarrow c=4$, which allows to add new pupils in class 4 in the real world. Then
$r_{1}$ is risky w.r.t.~$Q$, because pupils in class 4 can potentially
also live in Bolzano. Note that without integrity constraints, actually
most updates to the same relation will be risky: if we do not have
keys in the database, a pupil could live both in Bolzano and Merano
and hence an effect $r_{2}=\pupil(n,c,s)\leftsquigarrow\livesin(n,\Merano)$
would be risky w.r.t.~$Q$, too. If there is a key defined over the first
attribute of $\livesin$, then $r_{2}$ would not be risky, because
adding pupils that live in Merano would not influence the completeness
of pupils that only live in Bolzano.
\end{example}
We say that a real-world effect $r$ that is risky w.r.t.~a query $Q$
is \emph{repaired} by a set of copy effects $\{c_{2},\ldots,c_{n}\}$, if
for any sequence of databases $(\id{D_{1}},\id{D_{2}})$ that conforms
to $r$ it holds that $Q(\id{D_{2}})=Q(\id{D_{1}}\cup\mathit{\mathit{copy}}{}_{c_{1}\ldots c_{n}}(\id{D_{2}}))$.
Intuitively, this means that whenever we introduce new facts via $r$
and apply the copy effects afterwards, all new facts that can
change the query result are also copied into the information system.
\begin{proposition}[Repairing]
Consider the query $\query{Q}{R_1(\tpl t_1),\ldots R_n(\tpl t_n),M}$, let $\tpl v=\var(Q)$, a real-world effect
$R(\tpl x, \tpl y)\leftsquigarrow G_{1}(\tpl x,\tpl z_{1})$
and a set of copy effects
$\{c_{2},\ldots,c_{m}\}$. Then $r$ is repaired
by $\{c_{2},\ldots,c_{m}\}$ if and only if the following formula
is valid:
\[
\forall\tpl x,\tpl y\!:\ \biggl(\Bigl(\exists\tpl z_{1},\tpl v\!:\ (G_{1}(\tpl x,\tpl z_{1})
\wedge 
\bigwedge_{i=1\ldots n}\!R_{i}(\tpl{t}_{i})
\wedge 
M
\wedge 
\bigvee_{R_{i}=R}(\tpl x,\tpl y)=\tpl{t}_{i}\Bigr)
\ \ \Longrightarrow \ \ 
\bigvee_{j=2\ldots m} \exists \tpl z_j\!:\ G_j(\tpl x, \tpl z_j \biggr)
\]
\end{proposition}
\begin{fullversion}{
\begin{proof}
"$\Leftarrow :$"
Straightforward. If the formula is valid, it implies that any fact $R(\tpl x)$ that is introduced by the real-world effect $r$ and which can change the result of $Q$ also satisfies the condition of some copy effect and hence will be copied.

"$\Rightarrow :$"
Suppose the formula is not valid. Then there exists a fact $R(\tpl x)$ which satisfies the condition of the implication (so $R(\tpl x)$ can both conform to $r$ and change the result of $Q$) but not the consequence (it is not copied by any copy effect).
Thus, we can create a pair $(\id D_1, \id D_2)$ of databases as before as $\id D_1 = G_1(\tpl x, \tpl y) \cup \set{ \bigwedge_{i=1\ldots n,i\neq k} R_i(\tpl t_i)}$ and $\id D_2 = \id D_1 \cup \set{R_k(\tpl t_k)}$ which proves that $Q(\id D_2)\neq Q(\id D_1 \cup \mathit{copy}_{c_1,\ldots c_m}(\id D_2)$. \qed
\end{proof}
}\end{fullversion}
This implication can be translated into a problem of query
containment, a well-studied topic in database theory \cite{meyden-Complexity_querying_ordered_domains-pods,rosati:2003:containment-keys-and-foreign-keys,Sagiv:Yannakakis-Containment-VLDB,Razniewski:Nutt-Compl:of:Queries-VLDB11}.
\global\long\def\ucont#1{\mathit{ContU}(#1)}
\global\long\def\ent#1{\mathit{EntC}(#1)}
In particular, for a query $\query{Q(\tpl z)}{R_{1}(\tpl{t}_{1}),\ldots,R_{n}(\tpl{t}_{n})}$,
we define the atom-projection of $Q$ on the $i$-th atom as $\query{Q^\pi_{i}(\tpl x)}{R_{1}(\tpl{t}_{1}),\ldots,R_{n}(\tpl{t}_{n}),\tpl x=\tpl{t}_{i}}$.
Then, for a query $Q$ and a relation $R$, we define the $R$-projection
of $Q$, written $Q^{R}$, as the union of all the atom-projections
of atoms that use the relation symbol $R$, that is, $\bigcup_{R_{i}=R}Q^\pi_{i}$.
For a real-world effect $r=R(\tpl x, \tpl y)\leftsquigarrow G(\tpl x,\tpl z)$, we define its associated query $P_r$ as $\query{P_r(\tpl x, \tpl y)}{R(\tpl x, \tpl y), G(\tpl x, \tpl z)}$.
\begin{corollary}[Repairing and query containment]
\label{cor:repairing-if-containment}Let $Q$ be a query, $\alpha=a_1,\ldots a_n$ be an action sequence,
$a_{i}$ be an action with a risky real-world effect $r$, and $\{c_{1},\ldots,c_{m}\}$
be the set of all copy effects of the actions $a_{i+1}\ldots a_{n}$.

Then $r$ is repaired, if and only if it holds that $P_{r}\cap Q^{R}\subseteq P_{c_{1}}\cup\ldots\cup P_{c_{m}}$.
\end{corollary}
Intuitively, the corollary says that a risky effect $r$ is repaired, if all data that is introduced by $r$ that can potentially change the result of the query $Q$ are guaranteed to be copied into the information system database by the copy effects $c_1$ to $c_n$.

The corollary holds because of the direct correspondence between conjunctive queries and relational calculus \cite{foundations_of_dbs}.

We arrive at a result for characterizing query completeness wrt.\ an action sequence:
\begin{lemma}
[Action sequence completeness]Let $\alpha$ be
an action sequence and $Q$ be a query. Then $\alpha\models\compl Q$
if and only if all risky effects in $\alpha$ are repaired.\label{lem:incompleteness-if-unhealed-rw-action}\end{lemma}
\begin{proof}
``$\Leftarrow$'': Assume that all risky real-world effects in $\alpha$ are repaired in $\alpha$. Then by Lemma \ref{cor:repairing-if-containment} any fact introduced by a real-world effect $r$ which can potentially also
influence the satisfaction of $\compl{Q}$ also satisfies the condition of
some later copy effect, and hence it is eventually copied into some
$\av{D_{j}}$ and hence it also appears in $\av{D_{n}}$, which implies
that $C$ is satisfied over $(\id{D_{n}},\av{D_{n}})$.

``$\Rightarrow$'': Assume the repairing does not hold for some risky
effect $r$ of an action $a_{i}\in\alpha$. Then by Lemma \ref{cor:repairing-if-containment},
since the containment does not hold, there exists a database $D$
with a fact $R(t)$ that is in $Q_{r}\cap Q^{R}(D)$ but not in $Q_{c_{i+1}}\cup\ldots\cup Q_{c_{n}}(D)$.
Then, we can create a development $\id{D_{0}},\ldots,\id{D_{n}}$
of $\alpha$ as $\id{D_{0}},\ldots,\id{D_{i-1}}=D\setminus\{R(t)\}$
and $\id{D_{i}},\ldots,\id{D_{n}}=D$. Its trace is $\av{D_{0}},\ldots,\av{D_{n}}=D\setminus\{R(t)\}$,
because since the containment does not hold, for none of the copy
effects in the following actions its guard evaluates to true for the
fact $R(t)$ and hence $R(t)$ is never copied into the information system
database. But since $R(t)$ is in $Q^{R}(D)$, query completeness
for $Q$ is not satisfied over $(\id{D_{n}},\av{D_{n}})$ and hence
$\alpha\not\models\compl Q$. \qed
\end{proof}
%
Before discussing complexity results in Section \ref{sub:complexity},
we show that completeness entailment over action sequences and containment
of unions of queries have the same complexity.\global\long\def\L{\mathcal{L}}
A query language is defined by the operations that it allows. Common sublanguages of conjunctive queries are, e.g., queries without arithmetic comparisons (so-called relational queries), or queries without repeated relation symbols (so-called linear queries).

For a query language $\L$, we call $\ent{\L}$ the problem of deciding whether an action sequence $\alpha$ entails completeness of a query $Q$, where $Q$ and the real-world effects and the copy effects in $\alpha$ are formulated in language $\L$. Also, we call $\ucont{\L}$ the problem
of deciding whether a query is contained in a union of queries, where all are formulated in the language $\L$.
\begin{theorem}
Let $\L$ be a query languages. Then $\ent{\L}$
and $\ucont{\L}$ can be reduced to each other in linear
time.\label{thm:containment=complentailment}\end{theorem}
\begin{proof}
``$\Rightarrow$'': Consider the characterization shown in Lemma \ref{lem:incompleteness-if-unhealed-rw-action}.
For a fixed action sequence, the number of containment checks is the
same as the number of the real-world effects of the action sequence and thus linear.

``$\Leftarrow$'': Consider a containment problem $Q_{0}\subseteq Q_{1}\cup\ldots\cup Q_{n}$,
for queries in a language $\L$.
Then we can construct a QATS $\qats=(S,s_{0},A,E,\re,\ce)$ over the schema
of the queries together with a new relation $R$ with the same arity
as the queries where $S=\{s_{0},s_{1},s_{2}\}$, $A=\{a_{1},a_{2}\},\re(a_{1})=\{\id R(\tpl x)\leftsquigarrow Q_{0}(\tpl x)\}$
and $\ce(a_{2})=\bigcup_{i=1\ldots n}\{Q_{i}(\tpl x)\rightarrow\av R(\tpl x)\}$.
Now, the action sequence $a_{1},a_{2}$ satisfies a query completeness
for a query $\query{Q'(\tpl x)}{R(\tpl x)}$ exactly if $Q_{0}$ is
contained in the union of the queries $Q_{1}$ to $Q_{n}$, because
only in this case the real-world effect at action $a_{1}$ cannot
introduce any facts into $\id{D_{1}}$ of a development of $a_{1},a_{2}$,
which are not copied into $\av{D_{2}}$ by one of the effects of the
action $a_{2}$. \qed
\end{proof}
We discuss the complexity of query containment and hence of completeness
entailment over action sequences more in detail in Section \ref{sub:complexity}.

So far, we have shown how query completeness over a path can be checked.
To verify completeness in a specific state, we have to consider
all paths to that state, which makes the analysis more difficult.
We first introduce a lemma that allows to remove repeated actions
in an action sequence:
\begin{lemma}
[Duplicate removal]Let $\alpha=\alpha_{1},\tilde{a},\alpha_{2},\tilde{a},\alpha_{3}$
be an action sequence with $\tilde{a}$ as repeated action and let
$Q$ be a query. Then $\alpha$ satisfies $\compl Q$ if and only
if $\alpha'=\alpha_{1},\alpha_{2},\tilde{a},\alpha_{3}$ satisfies
$\compl Q$. \label{lem:repeated-transitions-can-be-ignored}\end{lemma}
\begin{fullversion}{
\begin{proof}
"$\Rightarrow$":
Suppose $\alpha$ satisfies $\compl Q$. Then, by Prop.\ \ref{lem:incompleteness-if-unhealed-rw-action}, all risky real-world effects of the actions in $\alpha$ are repaired. Let $a_r$ be an action in $\alpha$ that contains a risky real-world effect $r$. Thus, there must exist a set of actions $A_c$ in $\alpha$ that follows $a_r$ and contains copy effects that repair $r$. Suppose $A_c$ contains the first occurrence of $\tilde a$. Then, this first occurrence of $\tilde a$  can also replaced by the second occurrence of $\tilde a$ and then the modified set of actions also appears after $a_r$ in $\alpha'$.

"$\Leftarrow$":
Suppose $\alpha'$ satisfies $\compl Q$. Then, also $\alpha$ satisfies $\compl Q$ because adding the action $\tilde a$ earlier cannot influence query completeness: Since by assumption each risky real-world effect of the second occurrence of $\tilde a$ is repaired by some set of actions $A_c$ that follows $\tilde a$, the same set $A_c$ also repairs each risky real-world effect of the first occurrence of $\tilde a$.
\qed
\end{proof}
}\end{fullversion}
The lemma shows that our formalism can deal with cycles. While cycles imply the existence of sequences of arbitrary length, the lemma shows that we only need to consider sequences where each action occurs at most once. Intuitively, it is sufficient to check each cycle only once. Based on this lemma, we define the \emph{normal action sequence} of
a path $\pi$ as the action sequence of $\pi$ in which for all
repeated actions all but the last occurrence are removed.

\begin{proposition}
[Normal action sequences]Let $\qats=(T,\re,\ce)$ be a QATS, $\Pi$
be the set of all paths of $\qats$ and $Q$ be a query. Then
\begin{compactenum}
\item for each path $\pi\in\Pi$, its normal action sequence has at most
the length $\mid A\mid$,
\item there are at most $\Sigma_{k=1}^{\mid A\mid}\frac{\mid A\mid!}{(\mid A\mid-k)!}<(\mid \! A \! \mid+1)!$
different normal forms of paths,
\item for each path $\pi\in\Pi$, it holds that $\pi\models\compl Q$ if
its normal action sequence $\alpha'$ satisfies $\compl Q$.
\end{compactenum}
\end{proposition}
The first two items hold because normal action sequences do not contain
actions twice. The third item holds because of Lemma \ref{lem:repeated-transitions-can-be-ignored},
which allows to remove all but the last occurrence of an action in
an action sequence without changing query completeness satisfaction.

Before arriving at
the main result, we need to show that deciding whether a given normal action sequence can actually be realized by a path is easy:
\begin{proposition}
Given a QATS $\qats$, a state $s$ and a normal action sequence $\alpha$.
Then, deciding whether there exists a path $\pi$ that has $\alpha$
as its normal action sequence and that ends in $s$ can be done in
polynomial time.
\label{prop:checking-for-action-sequences-path-existence}
\end{proposition}
\global\long\def\reach{\mathit{reach}}
The reason for this proposition is that given a normal action sequence
$\alpha=a_{1},\ldots,a_{n}$, one just needs to calculate the states
reachable from $s_{0}$ via the concatenated expression $(a_{1},\ldots,a_{n})^{+},
(a_{2},\ldots,a_{n})^{+},\ldots,(a_{n-1},a_{n})^{+},(a_{n})^{+}$.
This expression stands exactly for all action sequences with $\alpha$
as normal sequence, because it allows repeated actions before their
last occurrence in $\alpha$. Calculating the states that are reachable
via this expression can be done in polynomial time, because the reachable
states $S_{n}^{\reach}$ can be calculated iteratively for each component
$(a_{i},\ldots,a_{n})^+$ as $S_{i}^{\reach}$ from the reachable states
$S_{i-1}^{\reach}$ until the previous component $(a_{i-1},\ldots,a_{n})^+$
by taking all states that are reachable from a state in $S_{i-1}^{\reach}$
via one or several actions in $\{a_{i},\ldots,a_{n}\}$, which can
be done with a linear-time graph traversal such as breadth-first or depth-first search. Since
there are only $n$ such components, the overall algorithm works in
polynomial time.

\begin{theorem}
Given a QATS $\qats$ and a query
$Q$, both formulated in a query  language $\L$, checking ``$s\not\models\compl Q$?'' can be done using a nondeterministic polynomial-time Turing machine with a $\ucont{\L}$-oracle.
\label{thm:final-characterization-design-time}\end{theorem}

\begin{proof}
If $s\not\models\compl{Q}$, one can guess a normal action sequence $\alpha$, check by Prop.~\ref{prop:checking-for-action-sequences-path-existence} in polynomial time that there exists a path $\pi$ from $s_0$ to $s$ with $\alpha$ as normal action sequence, and by Thm.~\ref{thm:containment=complentailment}  verify using the $\ucont{\L}$-oracle that $\alpha$ does not satisfy $\compl{Q}$. \qed
\end{proof}

We discuss the complexity of this problem in Section~\ref{sub:complexity}

\subsection{Runtime Verification}

Taking into account the concrete activities that were carried out
within a process can allow more conclusions about completeness. As
an example, consider that the secretary in a large school can perform
two activities regarding the enrollments, either he/she can sign enrollment
applications (which means that the enrollments become legally valid),
or he/she can record the signed enrollments that are not yet recorded
in the database. For simplicity we assume that the secretary batches
the tasks and performs only one of the activities per day. A visualization
of this process is shown in Fig. \ref{figure:runtime-verification-example-bpmn}.
Considering only the process we cannot draw any conclusions about
the completeness of the enrollment data, because if the secretary
chose the first activity, then data will be missing, however if the
secretary chose the second activity, then not. If however we have
the information that the secretary performed the second activity,
then we can conclude that the number of the currently valid enrollments
is also complete in the information system.

\begin{figure}[t]
\centering
\includegraphics[scale=1]{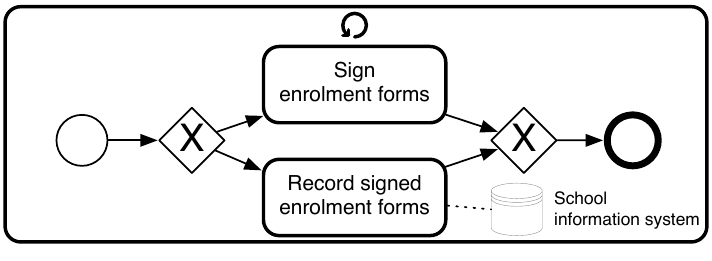}
\caption{Simplified BPMN process for the everyday activity of a
  secretary in a school}
\label{figure:runtime-verification-example-bpmn}
\end{figure}

Formally, in runtime verification we are given a path $\pi=t_{1},\ldots,t_{n}$
that was executed so far and a query $Q$. Again the problem is to
check whether completeness holds in the current state, that is, whether
all developments of $\pi$ satisfy $\compl Q$. 
\begin{corollary}
Let $\pi$ be a path in a QATS and $Q$ be a query, such that both $Q$ and the real-world effects and the copy effects in the actions of $\pi$ are formulated in a query language $\L$. Then ``$\pi\models\compl{Q}$?'' and $\ucont{\L}$ can be reduced to each other in linear time.
\end{corollary}

The corollary follows directly from Theorem \ref{thm:containment=complentailment}
and the fact that a path satisfies completeness if and only if its action sequence satisfies completeness.

Runtime verification becomes more complex when also the current, concrete
state of the information system database is explicitly taken into account. Given the current state $D$ of the database,
the problem is then to check whether all the developments of $\pi$
in which $\av{D_{n}}=D$ holds satisfy $\compl Q$.
In this case repairing of all risky actions is a sufficient but not
a necessary condition for completeness:

\begin{example}
\global\long\def\resident{\mathit{resident}}
Consider a path $(s_{0},a_{1},s_{1}),(s_{1},a_{2},s_{2})$, where
action $a_{1}$ is annotated with the copy effect $\id{\req}(n,s)\rightarrow\av{\req}(n,s)$, action $a_{2}$ with the real-world effect $\id{\pupil}(n,c,s)\leftsquigarrow\id{\req}(n,s)$,
 a database $\av{D}_2$ that is empty, and consider a query $\query{Q(n)}{\pupil(n,c,s),}$
$\req(n,s)$.
Then, the query result over $\av{D}_2$ is empty. Since the relation $\req$
was copied before, and is empty now, the query result over any real-world
database must be empty too, and therefore $\compl Q$ holds. Note
that this cannot be concluded with the techniques introduced in this work, as the
real-world effect of action $a_{2}$ is risky and is not repaired.
\end{example}
The complexity of runtime verification w.r.t.~a concrete database
instance is still open.

\subsection{Dimension Analysis}

When at a certain timepoint a query is not found to be complete,
for example because the deadline for the submissions of the enrollments
from the schools to the central school administration is not yet over,
it becomes interesting to know which parts of the answer are already
complete. 
\begin{example}
Consider that on the 10th of April, the schools ``Hofer''
and ``Da Vinci'' have confirmed that they have already submitted
all their enrollments, while ``Max Valier'' and ``Gherdena'' have
entered some but not all enrollments, and other schools did not enter
any enrollments so far. Then the result of a query asking for the
number of pupils per school would look as in Fig. \ref{figure:visualization-dimension-analysis}
(left table), which does not tell anything about the trustworthiness
of the result. If one includes the information from the process, one
could highlight that the data for the former two schools is already
complete, and that there can also be additional schools in the query
result which did not submit any data so far (see right table in Fig.
\ref{figure:visualization-dimension-analysis}).\label{example:dimension-analysis}
\end{example}
\begin{figure}[t]
\begin{centering}
\includegraphics[width=.8\textwidth]{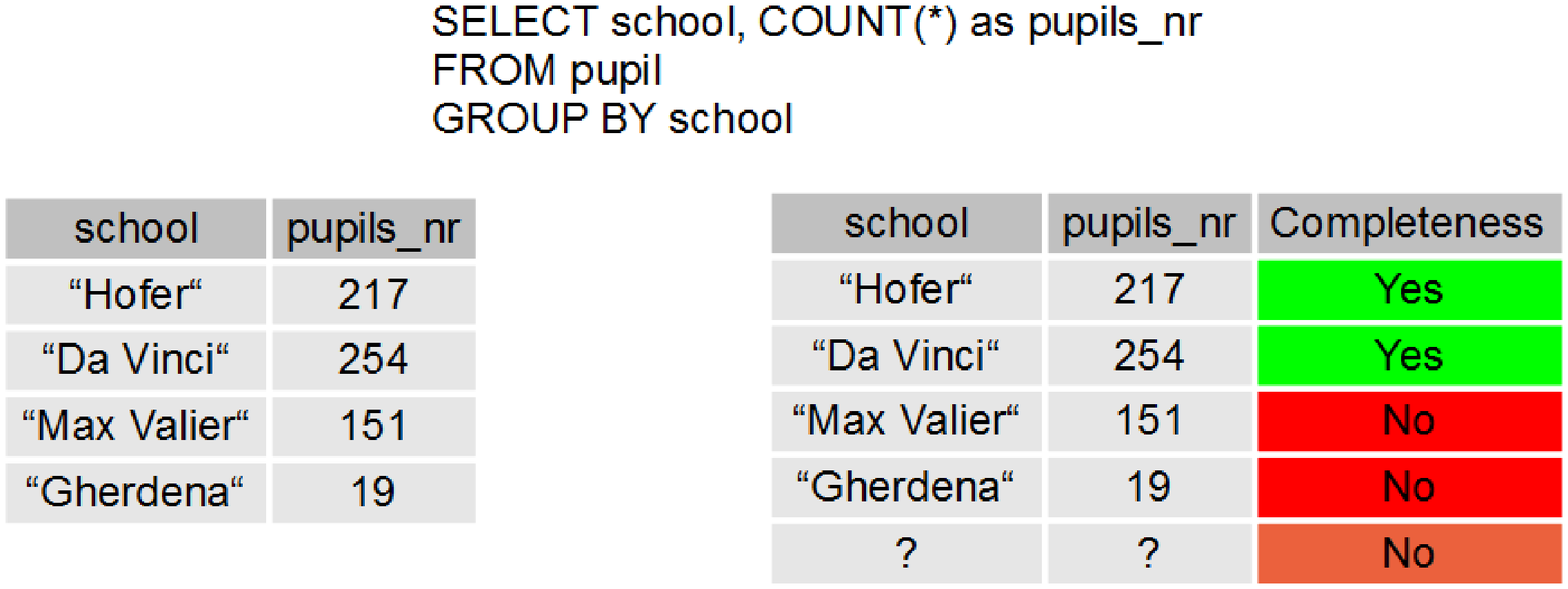}
\par\end{centering}

\caption{Visualization of the dimension analysis of Example \ref{example:dimension-analysis}.}
\label{figure:visualization-dimension-analysis}
\end{figure}

Formally, for a query $Q$ a dimension is a set of distinguished variables
of $Q$. Originally, dimension analysis was meant especially for the
arguments of a GROUP BY expression in a query, however it can also
be used with other distinguished variables of a query. Assume a query
$\query{Q(\tpl x)}{B(\tpl x,\tpl y)}$ cannot be guaranteed to be
complete in a specific state of a process. For a dimension $\tpl w\subseteq\tpl x$,
the analysis can be done as follows:
\vspace{3pt}
\begin{compactenum}
\item Calculate the result of $\query{Q'(\tpl w)}{B(\tpl x,\tpl y)}$ over
$\av D$.
\item For each tuple $\tpl c$ in $Q'(\av D)$, check whether $s,\av D\models\compl{Q[\tpl w/\tpl c]}$.
This tells whether the query is complete for the values $\tpl c$
of the dimension.
\item To check whether further values are possible, one has to guess a new value $\tpl{c}_{new}$ for the dimension and show that $Q[\tpl w/\tpl{c}_{new}]$ is not complete in the current state. For the guess one has to consider only the constants in the database plus a fixed set of new constants, hence the number of possible guesses is polynomial for a fixed dimension $\tpl v$.
\end{compactenum}
\vspace{3pt}
Step 2 corresponds to deciding for each tuple with a certain value
in $Q(\av D)$, whether it is complete or not (color red or green
in Fig. \ref{figure:visualization-dimension-analysis}, right table),
Step 3 to deciding whether there can be additional values (bottom
row in Fig. \ref{figure:visualization-dimension-analysis}, right
table).

\subsection{Complexity of Completeness Verification}

\label{sub:complexity}

In the previous sections we have seen that completeness verification can be solved using query containment.
Query containment is a problem that
has been studied extensively in database research. Basically, it is
the problem to decide, given two queries, whether the first is more
specific than the second.
The results follow from Theorem~\ref{thm:containment=complentailment}
and~\ref{thm:final-characterization-design-time}, and are summarized in Figure~\ref{figure:complexity-results-table}. We distinguish between the problem of runtime verification, which has the same complexity as query containment, and design-time verification, which, in principle requires to solve query containment exponentially often. Notable however is that in most cases the complexity of runtime verification is not higher than the one of design-time verification.
The results on linear relational and linear conjunctive queries, i.e., conjunctive queries without
selfjoins and without or with comparisons, are borrowed from \cite{Razniewski:Nutt-Compl:of:Queries-VLDB11}. 
The result on relational queries is reported
in~\cite{Sagiv:Yannakakis-Containment-VLDB}, and that on conjunctive queries from~\cite{meyden-Complexity_querying_ordered_domains-pods}.
As for integrity constraints, the result for databases satisfying finite domain constraints is reported in~\cite{Razniewski:Nutt-Compl:of:Queries-VLDB11} and for databases satisfying keys and foreign keys in~\cite{rosati:2003:containment-keys-and-foreign-keys}.

\begin{figure}[t]
\begin{centering}
\begin{scriptsize}
\begin{tabular}{|>{\centering}m{4cm}||>{\centering}m{3.5cm}|>{\centering}m{2.5cm}|}
\hline 
Query/QATS language $\L$ & Complexity of $\ucont{\L}$ and $\ent{\L}$
``$(\pi\models\compl Q)$''? & Complexity of ``$s\models\compl Q$''?\tabularnewline
\hline 
\hline 
Linear relational queries & PTIME & in $\coNP$\tabularnewline
\hline 
Linear conjunctive queries & $\coNP$-complete & $\coNP$-complete\tabularnewline
\hline 
\noalign{\vskip0.05cm}
Relational conjunctive queries & NP-complete & in $\Pi_{2}^{P}$\tabularnewline[0.05cm]
\hline 
Relational conjunctive queries over databases with finite domains & $\Pi_{2}^{P}$-complete & $\Pi_{2}^{P}$-complete\tabularnewline
\hline 
\noalign{\vskip0.05cm}
Conjunctive queries with comparisons & $\Pi_{2}^{P}$-complete & $\Pi_{2}^{P}$-complete\tabularnewline[0.05cm]
\hline 
Relational conjunctive queries over databases with keys and foreign keys & in PSPACE & in PSPACE\tabularnewline
\hline 
\end{tabular}
\end{scriptsize}
\par\end{centering}

\caption{Complexity of design-time and runtime verification for different query languages.}
\label{figure:complexity-results-table}
\end{figure}

\section{Conclusion}

In this paper we have discussed that data completeness analysis should take into account the processes that manipulate the data. In particular, we have shown how process models can be annotated with effects that create data in the real world and effects that copy data from the real world into an information system. We have then shown how one can verify the completeness of queries over transition systems that represent the execution semantics of such processes. It was shown that the problem is closely related to the problem of query containment, and that more completeness can be derived if the run of the process is taken into account.

In this work we focussed on the process execution semantics in terms
of transition systems. The next step is to realize a demonstration
system to annotate high-level business process specification languages
(such as BPMN or YAWL), extract the underlying quality-aware
transition systems, and apply the techniques here presented to check
completeness. Also, we intend to face the open question of completeness verification at runtime taking into account the actual database instance.

 \subsubsection*{Acknowledgements }

 This work was partially supported by the ESF Project 2-299-2010 \textquotedbl{}SIS
 - Wir verbinden Menschen\textquotedbl{}, and by the EU Project FP7-257593 ACSI. We are thankful to Alin Deutsch
 for a visit that helped initiating this research, and to the anonymous reviewers for helpful comments.

\bibliographystyle{splncs03}
\bibliography{completeness}
\addcontentsline{toc}{section}{completeness}


\end{document}